\documentclass[10pt, reqno, oneside]{amsart}

\usepackage{enumerate, microtype, amssymb, addlines, tikz, mathtools, ragged2e}
\usepackage[longnamesfirst]{natbib}
\shortcites{dedeckeretal2007}
\defcitealias{ibragimovmueller2016}{Ibragimov-M\"uller}
\defcitealias{canayetal2014}{Canay-Romano-Shaikh}
\defcitealias{besteretal2014}{Bester-Conley-Hansen}
\defcitealias{conleytaber2011}{Conley-Taber}

\input{macros}

\theoremstyle{plain}
\newtheorem{theorem}{Theorem}[section]

\newtheorem{proposition}[theorem]{Proposition}

\theoremstyle{definition}
\newtheorem{algorithm}[theorem]{Algorithm}

\newtheorem{example}[theorem]{Example}
\theoremstyle{remark}
\newtheorem*{remarks}{Remarks}
\newtheorem*{remark}{Remark}

\newcommand{\ev}{{\mathord \mathrm{E}}}%
\newcommand{\prob}{{\mathord P}}%
\DeclareMathOperator{\id}{\mathit{id}}%
\DeclareMathOperator*{\var}{Var}%
\DeclareMathOperator*{\diag}{diag}%
\newcommand{\pto}{\mathchoice
	{\raisebox{.0em}{ $\overset{\mathrm{P}}{\to}$ }}
	{\raisebox{-.15em}{ $\overset{\raisebox{-.25em}{\scriptsize$\mathrm{P}$}}{\to}$ }}
	{}
	{}}
\newcommand{\wto}[1]{\mathchoice
	{\raisebox{.0em}{ $\overset{#1}{\leadsto}$ }}
	{\raisebox{-.15em}{ $\overset{\raisebox{-.25em}{\scriptsize$#1$}}{\leadsto}$ }}
	{}
	{}}
\newcommand{\xqed}[1]{%
  \leavevmode\unskip\penalty9999 \hbox{}\nobreak\hfill
  \quad\hbox{\ensuremath{#1}}}
\newcommand{\sqed}{\xqed{\square}}

\newcommand{\Sym}{\mathfrak{G}}
\newcommand{\G}{\mathfrak{G}}
\newcommand\mydots{\hbox to 1em{.\hss.\hss.}}
\newcommand{\ubar}[1]{\text{\b{$#1$}}}
\newcommand{\balpha}{\bar{\alpha}}
\newcommand{\smx}{\triangledown}

\usepackage{xcolor}
\definecolor{umorange}{HTML}{cc6600}
\definecolor{umblue}{HTML}{587abc}
\definecolor{umgrey}{HTML}{989c97}
\definecolor{umred}{HTML}{7a121c}
\definecolor{umgreen}{HTML}{83b2a8}

\usepackage[bookmarks=false,pdfpagemode=UseNone,pdftex,colorlinks,citecolor=black,%
            filecolor=black,linkcolor=black,urlcolor=black,%
            pdfauthor=Andreas~Hagemann,pdfstartview=FitH]{hyperref}

\begin{document}

\title[Inference with a single treated cluster]{Inference with a single treated cluster}
\author[A.~Hagemann]{Andreas Hagemann%\\ \\ (preliminary and incomplete, please do not distribute)
}
\address{Department of Economics, University of Michigan, 611 Tappan Ave, Ann Arbor, MI 48109, USA. Tel.: +1 (734) 764-2355. Fax: +1 (734) 764-2769}
\email{\href{mailto:hagem@umich.edu}{hagem@umich.edu}}
\urladdr{\href{https://umich.edu/~hagem}{umich.edu/~hagem}}
\date{\today. (First version on arXiv: October 9, 2020.)
}
\thanks{All errors are my own. Comments are welcome. I would like to thank Sarah Miller for useful discussions.
}

\begin{abstract}
I introduce a generic method for inference about a scalar parameter in research designs with a finite number of heterogeneous clusters where only a single cluster received treatment. This situation is commonplace in difference-in-differences estimation but the test developed here applies more generally. I show that the test controls size and has power under asymptotics where the number of observations within each cluster is large but the number of clusters is fixed. The test combines weighted, approximately Gaussian parameter estimates with a rearrangement procedure to obtain its critical values. The weights needed for most empirically relevant situations are tabulated in the paper. Calculation of the critical values is computationally simple and does not require simulation or resampling. The rearrangement test is highly robust to situations where some clusters are much more variable than others. Examples and an empirical application are provided. 
\vskip 1em \noindent
\emph{JEL classification}: C01, C22, C32\\ 
\emph{Keywords}: cluster-robust inference, difference in differences, two-way fixed effects, clustered data, dependence, heterogeneity
\end{abstract}

\maketitle\renewcommand{\bfseries}{\fontseries{b}\selectfont}

\section{Introduction}

Inference about the average effect of a binary treatment or policy intervention is often much more challenging than its estimation. For example, calculating a difference-in-differences estimate can be as simple as comparing the difference in average outcomes of individuals in a group before and after an intervention to the same differences in unaffected groups. The main challenge for inference is that individuals within each of these groups likely depend on one another in unobservable ways. Taking this dependence into account generally requires knowledge of an explicit ordering of the dependence structure within each group. While time-dependent data have a natural ordering, it may be difficult or impossible to credibly order cross-sectionally dependent data within states or villages. Researchers commonly try to sidestep this problem by splitting large groups into smaller clusters that are presumed to be independent in order to have access to standard inferential procedures based on cluster-robust standard errors or the bootstrap. Splitting states, villages, or other large groups into smaller clusters is often difficult to justify but necessary for most of the available inferential procedures because they achieve consistency by requiring the number of clusters to go to infinity. If a procedure is valid with a fixed number of clusters, it typically requires at least two treated clusters unless strong homogeneity conditions are satisfied. %For example, if the research design compares a group of residents of a state that expanded access to health insurance to similar groups of residents of neighboring states that did not expand access, then states generally have to be split into regions such as counties and inference is only robust to dependence within counties. 
Numerical evidence by \citet{bertrandetal2004}, \citet{mackinnonwebb2014}, and others suggests that ignoring dependence and heterogeneity may lead to heavily distorted inference in empirically relevant situations. In both cases, the actual size of the test can exceed its nominal level by several orders of magnitude, i.e., nonexistent effects are far too likely to show up as highly significant.

In this paper, I introduce an asymptotically valid method for inference with a single treated cluster that allows for heterogeneity of unknown form. The number of observations within each cluster is presumed to be large but the total number of clusters is fixed. The method, which I refer to as a \emph{rearrangement test}, applies to standard difference-in-differences estimation and other settings where treatment occurs in a single cluster and the treatment effect is identified by between-cluster comparisons. The key theoretical insight for the rearrangement test is that a mild restriction on some but not all of the heterogeneity in two samples of independent normal variables allows testing the equality of their means even if one sample consists of only a single observation. I prove that this is possible for empirically relevant levels of significance if the other sample consists of at least ten observations. The rearrangement test compares the data to a reordered version of itself after attaching a special weight to the sample with a single observation. The weights needed for most standard situations are tabulated in the paper and calculating additional weights is computationally simple. I also show that the weights remain approximately valid if the two samples of independent heterogeneous normal variables arise as a distributional limit. I exploit this result in the context of cluster-robust inference by constructing asymptotically normal cluster-level statistics to which the rearrangement test can be applied. The resulting test is consistent against all fixed alternatives to the null, powerful against $1/\sqrt{n}$ local alternatives, and does not require simulation or resampling.

Inference based on cluster-level estimates goes back at least to \citet{famamacbeth1973}. Their approach is generalized and formally justified by \citet{ibragimovmueller2010, ibragimovmueller2016}, who construct $t$ statistics from cluster-level estimates and show that these statistics can be compared to Student $t$ critical values. \citet{canayetal2014} obtain null distributions by permuting the signs of cluster-level statistics under symmetry assumptions. \citet{hagemann2019b} permutes cluster-level statistics directly but adjusts inference to control for the potential lack of exchangeability. All of these methods allow for a fixed number of large and heterogeneous clusters but require several treated clusters. The rearrangement test complements these methods because it relies on the same type of high-level condition on the cluster-level statistics but is valid with a single treated cluster. %The rearrangement test also differs from these methods in its construction and theoretical justification. 
Other methods that are valid with a fixed number of clusters are the tests of \citet{besteretal2014} and a cluster-robust version of the wild bootstrap \citep[see, e.g,][]{cameronetal2008, djogbenouetal2019} analyzed by \citet{canayetal2018}. However, these papers rely on strong homogeneity conditions across clusters that are not needed here.

Several approaches for inference have been developed specifically for difference-in-differences estimation. \citet{conleytaber2011} provide a method that is valid with a single treated cluster and infinitely many control clusters under strong independence and homogeneity conditions that justify an exchangeability argument. \citet{fermanpinto2019} extend this approach to situations where the form of heteroskedasticity is known exactly. Another extension by \citet{ferman2020} allows for spatial correlation while maintaining \citeauthor{conleytaber2011}'s exchangeability condition. The rearrangement test differs from these methods because it is not limited to models estimated by difference in differences, does not rely on exchangeability conditions, and allows for completely unknown forms of heterogeneity. Other approaches due to \citet{mackinnonwebb2019b, mackinnonwebb2019a} use randomization (permutation) inference for difference-in-differences estimation and other models with few treated clusters. They test ``sharp'' \citep{fisher1935} nulls under randomization hypotheses and asymptotics where the number of clusters is eventually infinite. In contrast, the present paper is able to test conventional nulls in a setting with finitely many clusters.
 
%Other methods for cluster-robust inference are surveyed in \citet{cameronmiller2014} and \citet{mackinnon2019}. These methods are designed for moderate to large numbers of clusters and typically rely on adjusting standard $t$ and $F$ tests in a linear regression model where the number of clusters is eventually infinite. \citet{donaldlang2007} derive standard error and degrees-of-freedom corrections under parametric distributional assumptions in a random effects model.  \citet{imbenskolesar2016} derive corrections to degrees of freedom and standard errors following the approach of \citet{bellmccaffrey2002}. \citet{carteretal2013} develop measures that can be used to determine degrees-of-freedom corrections. Important early work that recognizes the need for corrections includes \citet{kloek1981} and \citet{moulton1990}. Adjusted permutation inference differs fundamentally from all of these approaches because it keeps the number of clusters fixed, applies to a large number of models other than the linear regression model, and derives adjusted critical values from the data that automatically account for within-cluster dependence instead of correcting the degrees of freedom of standard critical values.

The remainder of the paper is organized as follows: Section~\ref{s:normal} proves several new results on normal random vectors with independent, heterogeneous entries after a specific transformation and introduces the rearrangement test. Section~\ref{s:singleclust} establishes the asymptotic validity of the test in the presence of finitely many heterogeneous clusters when only one cluster received treatment and discusses several examples. Section~\ref{s:montecarlo} illustrates the finite sample behavior of the new test in simulations and in data used by \citet{garthwaiteetal2014}, who analyze the effects of a large-scale disruption of public health insurance in Tennessee. Section~\ref{s:conc} concludes. The appendix contains auxiliary results and proofs.

I will use the following notation. $1\{A\}$ is an indicator function that equals one if $A$ is true and equals zero otherwise. Limits are as $n\to\infty$ unless noted otherwise and $\leadsto$ denotes convergence in distribution.

\section{Inference with heterogenous normal variables}\label{s:normal}
In this section, I construct a test for the equality of means of two samples of independent heterogeneous normal variables where one sample consists of only a single observation. The other sample has finitely many observations. I show that the test has power while controlling size (Theorem \ref{t:size}) and remains approximately valid if this two-sample problem characterizes the large sample distribution of a random vector of interest (Proposition \ref{t:approx}).

Consider $q$ independent variables $X_{0,1},\dots, X_{0,q}$ with $X_{0,k}\sim N(\mu_0,\sigma_k^2)$ for $1\leq k\leq q$. Independently, there is an additional variable $X_1 \sim N(\mu_1, \sigma^2)$. I interpret this as a two-sample problem with ``control'' sample $X_{0,1},\dots, X_{0,q}$ and ``treatment'' sample $X_1$, although all of the following still applies if these roles are reversed. The objective is to test the null hypothesis of equality of means,
\begin{equation*}%\label{eq:h0norm}
	H_0\colon \mu_1 = \mu_0,
\end{equation*}
without knowledge of $\mu_0, \sigma, \sigma_1,\dots,\sigma_q$ and without assuming that these quantities can be consistently estimated. I account for the uncertainty about $\mu_0$ by recentering the data $X = (X_1, X_{0,1},\dots, X_{0,q})$ with $\bar{X}_0 = q^{-1}\sum_{k=1}^q X_{0,k}$ to define
\begin{equation}\label{eq:xdef}
%X(\mu) = \bigl((1+w)(Y-\mu),(1-w)(Y-\mu), Z_1 - \mu,\dots, Z_q - \mu\bigr) \text{ with } X = X(\bar{Z}).
S(X,w) = \bigl((1+w)(X_1-\bar{X}_0),(1-w)(X_1-\bar{X}_0), X_{0,1} - \bar{X}_0,\dots, X_{0,q} - \bar{X}_0\bigr)
\end{equation}
for some known weight $w\in(0,1)$ that will be chosen shortly. If $X_1-\bar{X}_0 >0$, the $1+w$ increases $X_1-\bar{X}_0$ and $1-w$ decreases $X_1-\bar{X}_0$. If $X_1-\bar{X}_0 < 0$, these effects are reversed. The idea underlying the test is that if the decreased version of $X_1-\bar{X}_0$ is still large in comparison to $X_{0,1} - \bar{X}_0,\dots, X_{0,q} - \bar{X}_0$, then this size difference is unlikely to be only due to heterogeneity in $\sigma^2, \sigma_1^2,\dots,\sigma_q^2$ but provides evidence that $\mu_1$ and $\mu_0$ are in fact not equal. I show below that $w$ gives precise probabilistic control over this comparison. In particular, choosing $w$ appropriately allows me to construct a test whose size can be bounded at a predetermined significance level.

Before defining the test statistic, I first introduce some notation. For a given vector $s\in\mathbb{R}^d$, let $s_{(1)}\leq\cdots\leq  s_{(d)}$ be the ordered entries of $s$. Denote by $s\mapsto s^\smx = (s_{(d)},\dots,  s_{(1)})$ the operation of rearranging the components of $s$ from largest to smallest. %I will use this notation for reverse-order statistics and the ``$\smx$'' superscript throughout the paper and denote the operation of reordering the components of a vector $V\in\mathbb{R}^d$ from largest to smallest by $V^\smx = (V_{[1]}, \dots, V_{[d]})$. 
The test uses $S(X,w)$ and its rearranged version $S(X,w)^\smx$ in the difference-of-means statistic
\begin{equation}\label{eq:tdef}
	s = (s_1,\dots, s_{q+2}) \mapsto T(s) = \frac{s_1 + s_2}{2} -  \frac{1}{q}\sum_{k=1}^{q} s_{k+2}
\end{equation}
to define the test function 
\begin{equation}\label{eq:phidef}
\varphi (X, w) = 1\bigl\{ T\bigl(S(X,w)\bigr) = T\bigl(S(X,w)^\smx\bigr) \bigr\}.	
\end{equation}
The test, which I refer to as \emph{rearrangement test}, rejects if $\varphi (X,w) = 1$ and does not reject otherwise. As stated, the test is against the alternative of a positive treatment effect, $H_1\colon \mu_1 > \mu_0$. For a test against $H_1\colon \mu_1 < \mu_0$, simply use  $\varphi (-X, w)$. These alternatives can be combined to provide a two-sided test. I describe the exact implementation below equation \eqref{eq:sizealpha} ahead. Also note that the first difference of means in \eqref{eq:phidef} simplifies to $T(S(X,w))=X_1-\bar{X}_0$ but $T(S(X,w)^\smx)$ is in general a complicated function of $w$. 

Intuitively, the rearrangement test can be interpreted as a permutation test that treats $S = S(X,w)$ as if it were the data and uses the second largest permutation statistic of $T(S)$ as critical value $c$. If $T(S) > c$, then the only possibility left is that $T(S)$ equals its largest permutation statistic. For the difference of means $T(S)$, that statistic must be $T(S^\smx)$ and therefore $T(S) > c$ is equivalent to $\varphi (X,w) = 1$. Because $S$ is being permuted and not $X$, this also explains why it is sensible to write $T(S(X,w))$ instead of $X_1-\bar{X}_0$ in the definition of the test function \eqref{eq:phidef}. A classical permutation test would then use an exchangeability condition on $S$ to determine the size of the test. Even though the $S$ constructed here is far from exchangeable, I will show that this test has power while controlling size at a predetermined level. Instead of relying on exchangeability, the results here depend on the joint normality of $X$ combined with the location and scale invariance property $\varphi (X,w) = \varphi ((X-\mu_0 1_{q+1})/\sigma,w)$, where $1_{q+1}$ is a ($q+1$)-vector of ones. The location invariance is forced by the recentering of $X$ with $\bar{X}_0$ and effectively removes $\mu_0$ from the list of nuisance quantities. The scale invariance is ensured by the specific choices of $T$ and $\varphi$. It reduces the dimensionless unknowns $\sigma, \sigma_1,\dots,\sigma_q$ to the more tractable ratios $\sigma_1/\sigma,\dots,\sigma_q/\sigma$.

%By construction, $T(X(\mu))$ is location invariant with respect to the original data $(Y,Z_1, \dots, Z_q)$ in the sense that $T(X(\bar{Z})) = T(X(\mu))$ for every $\mu\in\mathbb{R}$, but this is not necessarily true for $T(X^\smx)$. 

I start with the analysis of size and power, and connect these results with the situation where $X = (X_1, X_{0,1},\dots, X_{0,q})$ is an asymptotic approximation later on. I assume that the variances $\sigma_k^2$ of the $X_{0,k}$, $1\leq k\leq q$, are bounded away from zero by some $\ubar{\sigma}^2>0$ for all but one $k$. This avoids a trivial and in practice easily recognizable situation where some of the $X_{0,k}$ are exactly equal. I also restrict the variance $\sigma^2$ of $X_1$ to be bounded above by some $\bar{\sigma}^2<\infty$ because letting $\sigma\to\infty$ in $\varphi(X,w)$ would have the same effect as setting all $\sigma_k^2$ equal to zero. Under the null hypothesis, the distribution of $\varphi (X,w)$ is then determined by the unknown value of \[\lambda \in \Lambda \coloneqq \{ (\mu_0, \sigma, \sigma_1,\dots, \sigma_q) \in\mathbb{R}\times (0,\infty)^{q+1} : \sigma \leq \bar{\sigma} \text{ and } \sigma_k \geq \ubar{\sigma} \text{ for all $k$ but one} \}.\] Under the alternative, the distribution of $\varphi (X,w)$ also depends on the treatment effect $\delta = \mu_1 - \mu_0$. I write $\ev_{\lambda, \delta}$ and $\prob_{\lambda, \delta}$ to emphasize this dependence but occasionally drop subscripts to prevent clutter. 

My strategy is to first bound the null rejection probability $\ev_{\lambda, 0} \varphi (X,w)$ uniformly in $\lambda\in\Lambda$ by a smooth function of the weight $w$. I can then find a $w$ to make the bound exactly equal to the desired significance level to guarantee size control. The bound is also a function of the number of control observations $q$ and the maximal relative heterogeneity $\varrho = \bar{\sigma}/\ubar{\sigma}$ of treated and untreated observations. The parameter $\varrho$ is user chosen and has a simple interpretation: it restricts how much more variable $X_1$ can be relative to the $X_{0,k}$ when one of the $\sigma_k$ equals zero and the remaining $\sigma_k$ are all equal to the lower limit $\ubar{\sigma}$. This is the worst-case scenario for the test because $X_1$ is then likely to be very large on accident in comparison to the $X_{0,k}$. In that scenario, a $\varrho$ of 5 simply means that the variance of $X_1$ can be up to $5^2 = 25$ times larger than the variances of all but one of the $X_{0,k}$ and ``infinitely more variable'' than the remaining $X_{0,k}$. %It is also worth emphasizing that $\varrho$ puts no limits on how much \emph{less} variable $X_1$ can be than $X_{0,1},\dots,X_{0,q}$. The reason is that independence of the $X_{0,k}$ makes a situation where the $X_{0,k}$ all simultaneously happen to be very small relative to $X_1$ quite unlikely even when $\ubar{\sigma}$ is large in comparison to $\bar{\sigma}$. As a consequence, $\varrho$ can be less than one. 
There are no restrictions on how much \emph{less} variable $X_1$ can be than $X_{0,1},\dots,X_{0,q}$ and, in particular, $\bar{\sigma}/\ubar{\sigma}$ can be less than one.

The following theorem is the main theoretical result of the paper. It establishes the existence of a size bound that is valid for a fixed number of control observations $q$ and fully accounts for the uncertainty about the parameters in $\Lambda$. The theorem also shows that the test has power against the alternative $H_1\colon \mu_1 > \mu_0$. Results in the other direction follow by considering $\ev_{\lambda,-\delta}\varphi(-X, w)$ instead of $\ev_{\lambda,\delta}\varphi(X, w)$. The discussion immediately below focuses on the implications of the theorem. I address some of its technical aspects towards the end of this section. Let $\Phi$ and $\phi$ denote the normal distribution and density functions, respectively. 
\begin{theorem}[Size and power]\label{t:size}
Let $X_1, X_{0,1,}\dots, X_{0,q}$ be independent with $X_1\sim N(\mu_0 + \delta, \sigma^2)$ and $X_{0,k}\sim N(\mu_0, \sigma_k^2)$ for $1\leq k\leq q$. If $\delta = 0$, then for all $w\in (0,1)$,
\begin{align} 
\sup_{\lambda \in \Lambda} \ev_{\lambda, 0} \varphi (X, w) \leq \xi_q(w,\varrho) \coloneqq \frac{1}{2^{q+1}} + &\int_{0}^{\infty} \Phi\bigl((1-w)  \varrho y\bigr)^{q-1} \phi(y) dy \label{eq:xidef} \\ &~+ \min_{t > 0} \biggl ( \Phi\Bigl(\sqrt{q-1} w t \Bigr)^{q-1} + 2\Phi(- qt ) \biggr). \nonumber
\end{align}
%At $w = 1$, $\sup_{\lambda \in \Lambda} \ev_{\lambda,1} \varphi (X(\mu_0))\leq (5 + 3/2^{q-1})/2^{q+1}$.
Furthermore, for every $\lambda \in \Lambda$ and $w\in (0,1)$, we have $\lim_{\delta \to \infty}\ev_{\lambda,\delta} \varphi (X, w)= 1$ and $\lim_{\delta \to \infty}\ev_{\lambda, \delta} \varphi (X, 1)= 0$.
\end{theorem}
The theorem implies that the rearrangement test controls size, i.e., \[ \sup_{\lambda \in \Lambda} \ev_{\lambda, 0} \varphi (X, w) \leq \alpha, \] whenever $q$, $w$, and $\varrho$ are such that $\xi_q(w,\varrho) \leq \alpha$ for the desired significance level $\alpha$. %I examine for which choices of $w$ and $\varrho$ this is feasible below Proposition \ref{p:lowerbound} (ahead) after first discussing several aspects of the theorem in detail.
The bound $\xi_q(w,\varrho)$ has several properties that make this possible. In particular, it is monotonically increasing in $\varrho$ and decreasing in $q$. The reason for the monotonicity is that if $X_1$ can be more variable than $X_{0,1},\dots,X_{0,q}$, then the burden of proof to show ``$\mu_1 > \mu_0$'' as opposed to ``$\mu_1 = \mu_0$ with a large realization of $X_1$'' becomes necessarily higher. A large $q$ can ameliorate this effect somewhat because it removes uncertainty about $\mu_0$. The bound also tends to be decreasing in $w \in [0,1]$ because the integral generally dominates the other components,  but can increase slightly in some situations. 
%\path[use as bounding box,fill=fillColor,fill opacity=0.00] (0,6ex) rectangle (361.35,220);
\begin{figure}
\centering
\resizebox{.85\textwidth}{!}{
\input{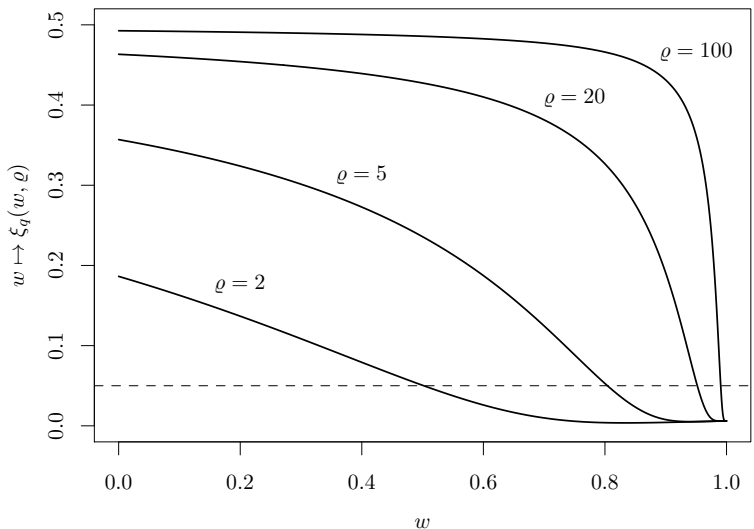}
}
\caption{Solid lines show the size bound $\xi_q(w, \varrho)$ at $q = 20$ control observations as a function of the weight $w$ for different values of the maximal heterogeneity~$\varrho$. The dashed line equals $.05$.}\label{f:fig1}
\end{figure}
This is illustrated in Figure~\ref{f:fig1}, where $w\mapsto \xi_q(w,\varrho)$ (solid lines) is essentially decreasing over the entire domain except for $\varrho=2$ and $w\geq .85$. Most importantly, it can be seen that $w\mapsto \xi_q(w,\varrho)$ decreases enough to dip below the desired significance level $\alpha = .05$ (dashed line) for all values of $\varrho$. As $q$ increases (not shown), $w\mapsto \xi_q(w,\varrho)$ is pushed towards zero but the shape of the function does not change meaningfully with $q$. The $w$ at which $\xi_q(w,\varrho)=\alpha$ is generally unique for most empirically relevant $\alpha$ and does not exist in some extreme situations. This can be seen in Figure~\ref{f:fig1}, where $w\mapsto \xi_q(w,\varrho)$ crosses $\alpha = .05$ only once for each $\varrho$ but, for example, $\xi_q(w,\varrho) = .6$ is never attained. 

Theorem~\ref{t:size} also provides information about the interplay between $w$ and the test under the alternative. In particular, it shows that the rearrangement test has power against $H_1 : \mu_1 > \mu_0$ for every $w\in (0,1)$ but the power declines sharply at $w=1$. I therefore explore the behavior of the test with $w$ near $1$ further in the following result. It provides a lower bound on the power of the test for fixed $\delta$. 
\begin{proposition}[Lower bound on power]\label{p:lowerbound} Let $X_1, X_{0,1},\dots, X_{0,q}$ be independent with $X_1\sim N(\mu_0 + \delta, \sigma^2)$ and $X_{0,k}\sim N(\mu_0, \sigma_k^2)$ for $1\leq k\leq q$. For every $w\in (0,1)$, $\sigma, \sigma_1,\dots, \sigma_q > 0$, and $\delta > 0$,
\[ \inf_{\mu_0\in\mathbb{R}}\ev_{\lambda, \delta} \varphi (X, w) \geq  2^q\sup_{t \geq 0}\Phi\biggl(\frac{\delta}{\sigma} - \frac{1+w}{1-w}t\biggr)\prod_{k=1}^q\Biggl(\Phi\biggl(\frac{\sigma}{\sigma_k}t\biggr) - 0.5\Biggr) \] The supremum is attained on $t\in (0,\infty)$. The right-hand side is strictly positive and converges to $1$ as $\delta \to \infty$.	
\end{proposition}
The bound shows that the test exhibits a standard relationship between the signal $\delta$ and the noise components $\sigma_1,\dots,\sigma_q$. Power is low if the signal relative to $\sigma$ is weak or the noise in the control group relative to $\sigma$ is strong. The latter relationship is in contrast to Theorem~\ref{t:size}, where small $\sigma_k$ relative to $\sigma$ were problematic. In addition, the bound also clarifies that $w$ dampens $\delta$ through the function $w\mapsto (1+w)/(1-w)$, which is arbitrarily large for $w$ sufficiently close to $1$. A $w$ very close to $1$ can therefore drown out a large treatment effect even if the noise coming from the control observations is mild. (The role of the supremum is simply to find the best possible balance for a given set of parameters.) It is also worth noting that the bound is tight enough to converge to $1$ as $\delta \to \infty$ and to  $0$ as $w \to 1$.

%Theorem \ref{t:size} implies that the rearrangement test controls size whenever $\xi_q$ as defined in \eqref{eq:xidef} does not exceed the significance level $\alpha$. For the observed $q$ and a given $\varrho$, this can be ensured by choosing a $w$ such that $\xi_q(w,\varrho) = \alpha$. 

Because the $w$ that satisfies $\xi_q(w,\varrho) = \alpha$ is not necessarily unique and because Proposition~\ref{p:lowerbound} suggests that power against the alternative $H_1 : \mu_1 > \mu_0$ for $w$ near one can be low, it is sensible to choose the smallest feasible $w$,  denoted by
\begin{equation}\label{eq:wdef}
w_q(\alpha,\varrho) = \inf\bigl\{ w \in (0,1) : \xi_q(w,\varrho) = \alpha \bigr\},
\end{equation}
in the definition of the rearrangement test function for a test of size $\alpha$,
\begin{equation}\label{eq:testfunalpha}
x\mapsto \varphi_\alpha(x) \coloneqq \varphi\bigl(x, w_q(\alpha,\varrho)\bigr).
\end{equation}
The test $\varphi_\alpha$ also depends on $\varrho$ but this is suppressed here to prevent clutter. %The $w$ at which $\xi_q(w,\varrho)$ equals $\alpha$ is generally unique for most empirically relevant $\alpha$ and does not exist in some extreme situations. This can be seen in Figure~\ref{f:fig1}, where $\alpha = .05$ (dashed line) crosses $w\mapsto \xi_q(w,\varrho)$ only once for each $\varrho$ but, for example, $\xi_q(w,\varrho) = .6$ is never attained. 
Table~\ref{tb:worstcasebounds} lists values of $w_q(\alpha,\varrho)$ for common choices of $\alpha$ as a function of $\varrho$ and $q$. They guarantee
\begin{equation}\label{eq:sizealpha}
\sup_{\lambda \in \Lambda} \ev_{\lambda,0} \varphi_\alpha (X) \leq \alpha.
\end{equation}
The list is not exhaustive and additional values can be easily calculated by numerical integration. %Because $w \mapsto \xi_q(w,\varrho)$ is continuous and generally either decreasing or bowl-shaped, this is computationally simple and $w_q(\alpha,\varrho)$ is already guaranteed to be attained as a minimum in \eqref{eq:wdef} if there are $w', w''$ such that $\xi_q(w',\varrho)<\alpha < \xi_q(w'',\varrho)$. 
An \texttt{R} command that performs the calculations can be found at \texttt{https://hgmn.github.io/rea}. 

\begin{table}\caption{Weights $w_q(\alpha,\varrho)$ as defined in \eqref{eq:wdef} that guarantee size control at $\alpha$ for a given maximal degree of heterogeneity $\varrho = \bar{\sigma}/\ubar{\sigma}$ for different values of $q$.}\label{tb:worstcasebounds}
\centering
\scalebox{.95}{
\begin{tabular}{cp{0cm}cp{0cm}cccccccccc}																												\hline
	&	&		&	&		\multicolumn{9}{c}{$q$}																				\\	\cline{5-13}
$\alpha$	&	&	$\bar{\sigma}/\ubar{\sigma}$	&	&		10	&		15	&		20	&		25	&	30	&	35	&	40	&	45	&	49	\\	\hline
.10	&	&	2	&	&	\it	.6333	&		.4010	&		.3294	&		.2829	&	.2475	&	.2188	&	.1948	&	.1742	&	.1562	\\	
	&	&	3	&	&			&		.6098	&		.5543	&		.5221	&	.4983	&	.4792	&	.4632	&	.4495	&	.4375	\\	
	&	&	4	&	&			&		.7127	&		.6669	&		.6418	&	.6238	&	.6094	&	.5974	&	.5871	&	.5781	\\	
	&	&	5	&	&			&	\it	.7732	&		.7344	&		.7137	&	.6991	&	.6876	&	.6779	&	.6697	&	.6625	\\	
	&	&	6	&	&			&	\it	.8129	&		.7792	&		.7615	&	.7493	&	.7396	&	.7316	&	.7248	&	.7188	\\	
	&	&	7	&	&			&	\it	.8409	&		.8111	&		.7957	&	.7851	&	.7768	&	.7700	&	.7641	&	.7590	\\	
	&	&	8	&	&			&	\it	.8616	&		.8350	&		.8213	&	.8120	&	.8048	&	.7987	&	.7936	&	.7891	\\	
	&	&	9	&	&			&	\it	.8776	&		.8536	&		.8413	&	.8329	&	.8265	&	.8211	&	.8165	&	.8125	\\	
																											\\	
.05	&	&	2	&	&			&	\it	.5752	&		.5020	&		.4615	&	.4318	&	.4081	&	.3884	&	.3715	&	.3568	\\	
	&	&	3	&	&			&	\it	.7287	&		.6703	&		.6414	&	.6213	&	.6054	&	.5923	&	.5810	&	.5712	\\	
	&	&	4	&	&			&	\it	.8024	&		.7541	&		.7314	&	.7161	&	.7041	&	.6942	&	.6858	&	.6784	\\	
	&	&	5	&	&			&	\it	.8450	&		.8042	&		.7854	&	.7729	&	.7633	&	.7554	&	.7486	&	.7428	\\	
	&	&	6	&	&			&	\it	.8727	&		.8374	&		.8213	&	.8108	&	.8028	&	.7962	&	.7905	&	.7856	\\	
	&	&	7	&	&			&	\it	.8921	&		.8610	&		.8469	&	.8379	&	.8310	&	.8253	&	.8205	&	.8163	\\	
	&	&	8	&	&			&	\it	.9064	&		.8786	&		.8661	&	.8582	&	.8521	&	.8471	&	.8429	&	.8392	\\	
	&	&	9	&	&			&	\it	.9173	&		.8923	&		.8811	&	.8739	&	.8685	&	.8641	&	.8604	&	.8571	\\	
																											\\	
.025	&	&	2	&	&			&	\it	.6981	&		.6049	&		.5656	&	.5387	&	.5175	&	.5001	&	.4852	&	.4723	\\	
	&	&	3	&	&			&			&		.7400	&		.7111	&	.6926	&	.6784	&	.6667	&	.6568	&	.6482	\\	
	&	&	4	&	&			&			&	\it	.8069	&		.7838	&	.7696	&	.7588	&	.7501	&	.7426	&	.7362	\\	
	&	&	5	&	&			&			&	\it	.8466	&		.8273	&	.8157	&	.8071	&	.8001	&	.7941	&	.7889	\\	
	&	&	6	&	&			&			&	\it	.8728	&		.8563	&	.8465	&	.8393	&	.8334	&	.8284	&	.8241	\\	
	&	&	7	&	&			&			&	\it	.8914	&		.8770	&	.8685	&	.8622	&	.8572	&	.8529	&	.8493	\\	
	&	&	8	&	&			&			&	\it	.9053	&		.8924	&	.8849	&	.8795	&	.8751	&	.8713	&	.8681	\\	
	&	&	9	&	&			&			&	\it	.9160	&		.9045	&	.8978	&	.8929	&	.8890	&	.8856	&	.8828	\\	
																											\\	
.01	&	&	2	&	&			&			&	\it	.6986	&		.6543	&	.6286	&	.6092	&	.5935	&	.5801	&	.5686	\\	
	&	&	3	&	&			&			&	\it	.8058	&		.7709	&	.7527	&	.7396	&	.7290	&	.7201	&	.7124	\\	
	&	&	4	&	&			&			&	\it	.8578	&		.8290	&	.8147	&	.8047	&	.7968	&	.7901	&	.7843	\\	
	&	&	5	&	&			&			&	\it	.8882	&		.8636	&	.8519	&	.8438	&	.8374	&	.8321	&	.8275	\\	
	&	&	6	&	&			&			&	\it	.9080	&		.8866	&	.8767	&	.8699	&	.8645	&	.8601	&	.8562	\\	
	&	&	7	&	&			&			&	\it	.9219	&		.9030	&	.8943	&	.8885	&	.8839	&	.8801	&	.8768	\\	
	&	&	8	&	&			&			&	\it	.9322	&	\it	.9153	&	.9076	&	.9024	&	.8984	&	.8951	&	.8922	\\	
	&	&	9	&	&			&			&	\it	.9401	&	\it	.9248	&	.9179	&	.9133	&	.9097	&	.9067	&	.9042	\\	
																											\\	
.005	&	&	2	&	&			&			&	\it	.7642	&		.7029	&	.6764	&	.6576	&	.6426	&	.6300	&	.6191	\\	
	&	&	3	&	&			&			&			&	\it	.8042	&	.7847	&	.7719	&	.7618	&	.7534	&	.7461	\\	
	&	&	4	&	&			&			&			&	\it	.8544	&	.8389	&	.8290	&	.8214	&	.8150	&	.8096	\\	
	&	&	5	&	&			&			&			&	\it	.8842	&	.8713	&	.8632	&	.8571	&	.8520	&	.8477	\\	
	&	&	6	&	&			&			&			&	\it	.9040	&	.8929	&	.8861	&	.8809	&	.8767	&	.8731	\\	
	&	&	7	&	&			&			&			&	\it	.9180	&	.9082	&	.9024	&	.8980	&	.8943	&	.8912	\\	
	&	&	8	&	&			&			&			&	\it	.9284	&	.9198	&	.9146	&	.9107	&	.9075	&	.9048	\\	
	&	&	9	&	&			&			&			&	\it	.9365	&	.9287	&	.9241	&	.9207	&	.9178	&	.9154	\\	\hline
\end{tabular}																												
}
\justify
\begin{adjustwidth}{.75em}{.75em}
\footnotesize\emph{Note:} Missing cells mean that the test is not recommended or not feasible. \emph{Italics} mean that the bound in \eqref{eq:xidef} is relatively loose. Upright numbers mean that the bound is nearly tight.
\end{adjustwidth}
\end{table}		

Table~\ref{tb:worstcasebounds} shows that the rearrangement test is available in a wide variety of situations depending on the desired significance level and tolerance for heterogeneity. For instance, a test with a 10\% significance level is already available with $q=10$ control observations. A 5\% level test becomes available at $q=15$, a 1\% level test at $q=20$, and for $q\geq 25$ there are essentially no restrictions to the level and underlying heterogeneity. This provides two avenues for implementation:
\begin{enumerate}
	\item Choose a desired maximal degree of heterogeneity $\varrho$ and make test decisions based on this choice. 
	\item Determine at which degree of maximal heterogeneity the null hypothesis can no longer be rejected. 
\end{enumerate}
The first option is similar in spirit to the ubiquitous \citet{staigerstock1997} rule of thumb for weak instruments, where an $F$ statistic larger than 10 corresponds to a tolerance for an at most 10\% bias (as defined in \citealp{stockyogo2001}) in the instrumental variables estimator relative to least squares. The second option takes the form of a ``robustness check.'' It has a meaningful interpretation because a result that is robust to a tenfold larger standard deviation in the treated observation relative to the control sample is more credible than a result that only survives a twofold difference in standard deviation. This second option leaves it up to the reader to decide whether the results are convincing.

The test decision itself is simple. Choose $w = w_q(\alpha,\varrho)$ from Table~\ref{tb:worstcasebounds} for a given number of control observations $q$, desired significance level $\alpha$, and maximal tolerance for heterogeneity, e.g., $\varrho = 2$. For this $w$, compute $S = S(X, w)$ as in \eqref{eq:xdef} and reorder the entries of $S$ from largest to smallest to obtain $S^\smx$. For an $\alpha$-level test of $\mu_1 = \mu_0$, reject in favor of $\mu_1 > \mu_0$ if $T(S) = T(S^\smx)$ as defined in \eqref{eq:tdef}. For a one-sided test with level $\alpha$ against $\mu_1 < \mu_0$, reject if $T(-S) = T((-S)^\smx)$. For a two-sided test with level $2\alpha$, reject in favor of $\mu_1 \neq \mu_0$ if either $T(S) = T(S^\smx)$ or $T(-S) = T((-S)^\smx)$. The ``robustness check'' increases $\varrho$ until the null hypothesis can no longer be rejected against the desired alternative. The test decision is monotonic in $\varrho$, i.e., if $\varrho' > \varrho$ lead to the same test decision, then the decision does not change for any value between $\varrho$ and $\varrho'$. An \texttt{R} command that implements the test and the robustness check for any choice of $\varrho$ is available at \href{https://hgmn.github.io/rea}{\texttt{https://hgmn.github.io/rea}}.

I now turn to a discussion of some technical aspects of the size bound $\xi_q(w,\varrho)$ that forms the theoretical underpinning for the rearrangement test. The bound, defined in \eqref{eq:xidef}, has three components with simple interpretations: The $1/2^{q+1}$ removes an unlikely event ($X_1 < \mu_0$, $X_{0,1} < \mu_0, \dots, X_{0,q} < \mu_0$ at the same time) from consideration. This forces a monotonicity property over the complement of this event and allows tightly bounding an oracle version of the problem where $\mu_0$ replaces $\bar{X}_0$ in \eqref{eq:xdef}. This bound is the integral in \eqref{eq:xidef}. The minimization problem then adjusts for the fact that the data are centered by $\bar{X}_0$ instead of the unknown $\mu_0$. The minimizer does not have closed form but is easily found numerically.\footnote{In particular, at $t = 1/q$, $\Phi(\sqrt{q-1} w t )^{q-1} + 2\Phi(- qt ) <  \Phi(1/\sqrt{q})^{q-1} + 2\Phi(- 1) < 1$ for $q > 2$. Because $\Phi(\sqrt{q-1} w t )^{q-1} + 2\Phi(- qt ) \geq 1$ at $t\in \{0,\infty \}$, the minimization problem always has an interior solution. This also implies that the bound as a whole is a smooth function of $w$ and $\varrho$.} Taken together, $\xi_q(w_q(\alpha, \varrho),\varrho)$ can therefore be roughly viewed as a tight bound for a high-probability event plus two small adjustments. I use Table~\ref{tb:worstcasebounds} to illustrate the relative size of these adjustments. In the table, empty cells correspond to situations where there is either no $w$ such that $\xi_q(w, \varrho) = \alpha$ or more than $\alpha/2$ of $\xi_q(w_q(\alpha, \varrho),\varrho)$ is taken up by the non-tight parts of the bound. Cells in \emph{italics} are settings where between $\alpha/2$ and $\alpha/10$ of the bound are taken up by the non-tight parts. The lack of tightness in the remaining cells is less than $\alpha/10$. For these cells $\sup_{\lambda \in \Lambda} \ev_{\lambda,0} \varphi_\alpha (X)$ approximately equals $\alpha$. As the table shows, $\xi_q(w_q(\alpha, \varrho),\varrho)$ is an essentially tight bound for $\sup_{\lambda \in \Lambda} \ev_{\lambda,0} \varphi_\alpha (X)$ for $q\geq 30$. The bound is also nearly tight for values of $q$ as small as 15 as long as $\varrho$ is not too large. I return to a discussion of this aspect of the rearrangement test in Example~\ref{ex:ctcomp} (ahead), where I illustrate the size of the test numerically.

Finally, before concluding this section, I show that the rearrangement test remains approximately valid for random vectors $X_n$ converging in distribution to the random vector $X = (X_1, X_{0,1},\dots, X_{0,q})$ described in Theorem \ref{t:size}. The reason is that $\ev \varphi(X_n, w)$ and $\ev \varphi(X, w)$ eventually coincide whenever $X$ has independent entries and a smoothly distributed first entry. The $X$ in Theorem~\ref{t:size} easily satisfies these conditions, which makes $\varphi_\alpha(X_n)$ asymptotically an $\alpha$-level test.
\begin{proposition}[Large sample approximation]\label{t:approx}
Let $X_1, X_{0,1},\dots, X_{0,q}$ be independent and let $X_1$ have a continuous distribution. If $X_n\leadsto X$, then $\ev \varphi (X_n, w) \to  \ev\varphi (X, w)$ for every $w\in(0,1)$.
\end{proposition}

I use Theorem \ref{t:size} and Proposition \ref{t:approx} in the next section to construct a simple method for inference with a single treated cluster. Section \ref{s:montecarlo} shows how the rearrangement test performs in Monte Carlo experiments.

\section{Inference with a single treated cluster}\label{s:singleclust}

In this section, I use a single high-level condition to extend the rearrangement test introduced in the previous section to a test about a scalar parameter in research designs with a finite number of large, heterogeneous clusters where only a single cluster received treatment. I then outline how these results can be applied in empirical practice. %Some of the discussion follows \citet{hagemann2019b}.

Suppose data from $q+1$ large clusters (e.g., states, industries, or villages observed over one or more time periods) are available. Data are dependent within clusters but independent across clusters. The exact form of dependence is unknown and not presumed to be estimable. An intervention took place during which one cluster received treatment and and $q$ clusters did not. The quantity of interest is a treatment effect or an object related to it that can be represented by a scalar parameter~$\delta$. Because the entire cluster was treated, this parameter is only identified up to a location shift $\theta_0$ within the treated cluster and therefore only the left-hand side of \[ \theta_1 = \theta_0 + \delta \] can be identified from this cluster. If the treated cluster would have behaved similarly to the untreated clusters in the absence of an intervention, then $\theta_0$ can be identified from each untreated cluster. Pairwise comparison then identifies $\delta$.

The identification strategy outlined in the preceding paragraph is the idea behind differences in differences---arguably the most popular identification strategy in modern empirical research---and a variety of other models. The goal of this section is to use the rearrangement test to provide a generic method for testing the hypothesis
%\[H_0\colon \theta_1 = \theta_0. \] 
\[H_0\colon \delta = 0, \] 
or, equivalently, $H_0\colon \theta_1 = \theta_0$. I achieve this by obtaining an estimate $\hat{\theta}_{1}$ of $\theta_1$ and estimates \smash{$\hat{\theta}_{0,1},\dots, \hat{\theta}_{0,q}$} of $\theta_0$ so that \[ \hat{\theta}_n = (\hat{\theta}_{1}, \hat{\theta}_{0,1} \dots, \hat{\theta}_{0,q}) \] is approximately a vector of independent but potentially heterogeneous normal variables that can be used as if it were the data vector $X$ from Section~\ref{s:normal}. 

The following example explains how to construct $\hat{\theta}_n$ in a simple situation. I discuss construction of $\hat{\theta}_n$ for difference in differences towards the end of this section.
\begin{example}[Regression with cluster-level treatment]\label{ex:clusterreg}
Consider a linear regression model
\begin{equation*}%\label{eq:regression}
Y_{i,k} = \theta_0 + \delta D_{k} + \beta_k' X_{i,k} + U_{i,k},
\end{equation*}
where $i$ indexes individuals within cluster $k$. There are $q+1$ clusters and individuals in cluster $k = q+1$ received treatment ($D_k=1$) but those in $1\leq k\leq q$ did not ($D_k=0$). The parameter of interest $\delta$ on the treatment indicator $D_k$ can be interpreted as an average treatment effect under suitable conditions. See, e.g., \citet{sloczynski2018, sloczynski2020} and references therein for a precise discussion. The regression may also include covariates $X_{i,k}$ that vary within each cluster and have coefficients $\beta_k$ that may vary across clusters. The condition $\ev(U_{i,k}\mid D_k, X_{i,k}) = 0$ identifies $\theta_1 = \theta_0 + \delta$ within the treated cluster and $\theta_0$ within the untreated clusters. The preceding display can then be written as
\begin{align*}
Y_{i,k} = 
\begin{cases}
\theta_0 + \beta_k' X_{i,k} + U_{i,k}, &1\leq k\leq q,\\ 
\theta_1 + \beta_k' X_{i,k} + U_{i,k}, &k = q + 1.
\end{cases}
\end{align*}
View these as $q+1$ separate regressions and use the least squares estimates of the constants $\theta_1$ and $\theta_0$ as the vector $\hat{\theta}_n = (\hat{\theta}_{1}, \hat{\theta}_{0,1} \dots, \hat{\theta}_{0,q})$ described above. \sqed
\end{example}

I will now show that the cluster-level statistics $\hat{\theta}_n$ can be used together with the results in the previous section to perform a consistent test as the sample size $n$ grows large.  The test is not limited to parameters estimated by least squares. Instead, consistency relies on the condition that a centered and scaled version of some estimate $\hat{\theta}_n$ converges to a $(q+1)$-dimensional normal distribution,
\begin{equation}\label{eq:jointconv}
\sqrt{n}\biggl(\frac{\hat{\theta}_{1} - \theta_1}{\sigma(\theta_1)},\frac{\hat{\theta}_{0,1} - \theta_0}{\sigma_{1}(\theta_0)}, \dots, \frac{\hat{\theta}_{0,q} - \theta_0}{\sigma_{q}(\theta_0)}\biggr) \wto{\theta} N(0, I_{q+1}),	
\end{equation}
where $\wto{\theta}$ denotes weak convergence under $\theta = (\theta_1,\theta_0)$. For fixed $\theta$, the display can be interpreted as $\sqrt{n}(\hat{\theta}_{1} - \theta_1, \dots,\hat{\theta}_{0,1} - \theta_0, \dots, \hat{\theta}_{0,q} - \theta_0)\leadsto N(0, \diag(\sigma, \sigma_1, \dots, \sigma_q))$ to include the case that one of the $\sigma_1,\dots, \sigma_q$ may be zero as in Theorem \ref{t:size}. 

A key feature of condition \eqref{eq:jointconv} is that the $\sigma$ and $\sigma_1,\dots, \sigma_q$ are not assumed to be known or estimable by the researcher. This is important for applications because consistent variance estimation generally requires knowledge of an explicit ordering of the dependence structure within each cluster. While time-dependent data are automatically ordered, it may be difficult or impossible to infer or credibly assume an ordering of the data within states or villages. In contrast, \eqref{eq:jointconv} can be established under weak (short-range) dependence conditions that only require \emph{existence} of a potentially unknown ordering for which the dependence of more distant units decays sufficiently fast. \citet{machkouriaetal2013} present convenient moment bounds and limit theorems for this situation. For more results in this direction, see also \citet{besteretal2014} and references therein. In general, the convergence in \eqref{eq:jointconv} also implicitly requires the number of observations in all clusters to grow with the sample size $n$. However, the clusters are not required to have similar or even identical sizes. Another noteworthy feature of condition \eqref{eq:jointconv} is the diagonal covariance matrix of the limiting distribution. It is the only independence condition that is imposed on the clusters.

I now show that under the joint convergence \eqref{eq:jointconv}, a rearrangement test that uses $\hat{\theta}_n$ is asymptotically of level $\alpha$ with a single treated cluster and a fixed number of control clusters. The test $\varphi_\alpha(\hat{\theta}_n)$, as defined in \eqref{eq:testfunalpha}, has power against all fixed alternatives $\theta_1 = \theta_0 + \delta$ with $\delta > 0$ and local alternatives $\theta_1 = \theta_0 + \delta/\sqrt{n}$ converging to the null. In the latter situation, $\theta_0$ is fixed and $\theta = (\theta_0 + \delta/\sqrt{n}, \theta_0)$ implicitly depends on $n$. The convergence in \eqref{eq:jointconv} is then a statement about an entire sequence $(\theta_0 + \delta/\sqrt{n}, \theta_0)$ instead of a single point. Results for alternatives with $\delta < 0$ follow from the same result by considering $\varphi_\alpha(-\hat{\theta}_n)$. These tests can be combined into a two-sided test that has power against fixed and local alternatives from either direction. Algorithm \ref{a:test} at the end of this section shows how this can be implemented.
\begin{theorem}[Consistency and local power] \label{t:consistency}
Suppose \eqref{eq:jointconv} holds with $\sigma^2 >0$ and at most one $\sigma_k = 0$. If $\theta_1 = \theta_0$, then \[ \lim_{n\to\infty}\ev \varphi_\alpha(\hat{\theta}_n) \leq \alpha,\qquad\text{every $\alpha,\varrho$ with $0 < w(\alpha,\varrho) < 1$,} \] and if $\theta_1 > \theta_0$, then $\ev \varphi_\alpha(\hat{\theta}_n)\to 1$. If \eqref{eq:jointconv} holds with $\theta = (\theta_0 + \delta/\sqrt{n},\theta_0)$ and the $\sigma, \sigma_1,\dots, \sigma_{q}$ are continuous and positive at $\theta_0$, then
\begin{equation*}\label{eq:asylocalpower}
\lim_{n\to\infty}\ev \varphi_\alpha(\hat{\theta}_n)  \geq 2^q\sup_{t \geq 0}\Phi\Biggl(\biggl(\frac{\delta}{\sigma(\theta_0)} - \frac{1+w_q(\alpha,\varrho)}{1-w_q(\alpha,\varrho)}t\biggr)\Biggr)\prod_{k=1}^q\Biggl(\Phi\biggl(\frac{\sigma(\theta_0)}{\sigma_k(\theta_0)}t\biggr) - 0.5\Biggr)>0.
\end{equation*}
\end{theorem}

\begin{remarks}
(i)~Because $\varphi_\alpha(\hat{\theta}_n) = 1$ if and only if $\varphi_\alpha(a(\hat{\theta}_n - \theta_0 1_{q+1})) = 1$, where $a> 0$ and $1_{q+1}$ is a $(q+1)$-vector of ones, the $\sqrt{n}$-rate in \eqref{eq:jointconv} and in the theorem can be replaced by any other rate as long as the asymptotic normal distribution in \eqref{eq:jointconv} is still attained. Several semiparametric or nonstandard estimators are therefore covered by the theorem.

(ii)~It is sometimes of interest in applications to test the null hypothesis $H_0\colon \theta_1 = \theta_0 + \gamma$ for a given $\gamma$. In that case, define $\Gamma = (\gamma 1\{k=1 \})_{1\leq k\leq q+1}$ and reject if $\varphi_\alpha(\hat{\theta}_n - \Gamma) = 1$. Replace $\theta_0$ by $\theta_0 + \gamma$ in Theorem \ref{t:consistency} and use part~(i) of this remark to see that this leads to a consistent test. \sqed 
\end{remarks}

I now discuss how the high-level condition \eqref{eq:jointconv} can be verified in an application. The specific example I use is difference-in-differences estimation but the arguments presented here apply more broadly. See also \citet{canayetal2014} and \citet{hagemann2019b} for similar types of arguments in other models. For simplicity, I focus on \eqref{eq:jointconv} under the null hypothesis $H_0 : \theta_1 = \theta_0$.

\begin{example}[Difference in differences]\label{ex:diffindiff}
Consider the panel model
\begin{equation}\label{eq:diffindiff}
Y_{i,t,k} = \theta_0 I_t + \delta I_t D_{k} + \beta_k' X_{i,t,k} + \zeta_{i,k}+ U_{i,t,k},
\end{equation}
where $i$ indexes individuals $i$ in unit $k\in\{1,\dots, q+1\}$ at time $t\in\{0,1\}$. Treatment occurred between periods $0$ and $1$. Right-hand side variables are a post-intervention indicator $I_t=1\{t= 1\}$, a treatment indicator $D_k$ that equals $1$ if unit $k$ ever received treatment, individual fixed effects $\zeta_{i,k}$, and other covariates $X_{i,t,k}$ that for every $k$ vary at least before or after the intervention. The collection of pre and post intervention data from unit $k$ forms the $k$-th cluster. Let $n_k$ be the number of individuals in cluster $k$ so that $n = 2\sum_{k=1}^{q+1} n_k$ is the total sample size. View each cluster as a separate regression and rewrite \eqref{eq:diffindiff} in first differences as
\begin{equation*}
\Delta Y_{i,k} = 
\begin{cases} 
\theta_0 + \beta_k' \Delta X_{i,k} + \Delta U_{i,k}, &1\leq k\leq q, \\ 
\theta_1 + \beta_k' \Delta X_{i,k} + \Delta U_{i,k}, &k = q + 1,
\end{cases}
\end{equation*}
where $\Delta Y_{i,k} = Y_{i,1,k} - Y_{i,0,k}$ and so on. Provided $\ev (\Delta U_{i,k}\mid \Delta X_{i,k}) = 0$, the data identify $\theta_1 = \theta_0 + \delta$ in a treated cluster and $\theta_0$ in an untreated cluster.  The least squares estimates $\hat{\theta}_1$ and $\hat{\theta}_{0,k}$ of the parameters $\theta_1$ and $\theta_0$ are suitable cluster-level estimates if $\hat{\theta}_n = (\hat{\theta}_{1}, \hat{\theta}_{0,1},\dots, \hat{\theta}_{n,q})$ satisfies condition \eqref{eq:jointconv}. 

In the absence of covariates (i.e., $\beta_k\equiv 0$), the centered and scaled least squares estimate in a control cluster under $H_0$ can be expressed as 
%\[ \sqrt{n}(\hat{\theta}_{0,k} - \theta_{1\{k\leq q_1\}}) = \sqrt{\frac{n}{n_{1,k}}} \frac{1}{\sqrt{n_{1,k}}}\sum_{t=n_{0,k}+ 1}^{n_{k}}U_{t,k} - \sqrt{\frac{n}{n_{0,k}}} \frac{1}{\sqrt{n_{0,k}}}\sum_{t=1}^{n_{0,k}}U_{t,k}. \] 
\[ \sqrt{n}(\hat{\theta}_{0,k} - \theta_0) = \biggl(\frac{n}{n_{k}}\biggr)^{1/2} n_{k}^{-1/2}\sum_{i=1}^{n_k}\Delta U_{i,k}. \] The same is true for $\sqrt{n}(\hat{\theta}_{1} - \theta_0)$ with $k=q+1$ on the right-hand side of the display.
If the number of individuals per cluster is large in the sense that $n/n_k \to c_k \in (0,\infty)$ for $1\leq k\leq q+1$, then condition \eqref{eq:jointconv} already holds if $n^{-1/2}(\sum_{i=1}^{n_k}U_{i,0,k},$ $\sum_{i=1}^{n_k}U_{i,1,k})$ is independent across $1\leq k\leq q+1$ and has a non-degenerate normal limiting distribution for each $k$. The latter condition can be ensured with a central limit theorem for spatially dependent data. See, e.g., \citet{jenischprucha2009} and \citet{machkouriaetal2013} for appropriate results. %\citet{dedeckeretal2007} for a comprehensive overview. 
If the number of individuals per cluster is small,  then Theorem~\ref{t:size} implies that the rearrangement test can still be applied under the assumption that $((U_{i,0,k})^T_{1\leq i\leq n_k}, (U_{i,1,k})^T_{1\leq i\leq n_k})$ is multivariate normal for $1\leq k\leq q+1$. This last condition may be strong but serves to illustrate that $\hat{\theta}_1$ and $\hat{\theta}_{0,k}$ need not even be consistent for the test to be valid.

Now consider pooled cross sections with $n_k$ individuals in period~$0$, $m_k$ individuals in period~$1$, and $\zeta_{i,k}\equiv \zeta_k$. The calculations in the preceding paragraph still apply with minor modifications. For period $1$, $n_k$ has to be replaced by $m_k$. The analysis is no longer in first differences but the underlying conditions are essentially identical as long as $n/n_k \to c_k\in (0,\infty)$ and $n/m_k \to c_k' \in (0,\infty)$ for $1\leq k\leq q+1$, where $n$ is the total sample size. If the number of individuals available post intervention $m  = \sum_{k=1}^{q+1} m_k$ is relatively small in the sense that $m/n_k \to 0$ and $m/m_k \to c_k' \in (0,\infty)$, the scale invariance discussed in the remarks below Theorem \ref{t:consistency} allows replacement of the $\sqrt{n}$ in \eqref{eq:jointconv} by $\sqrt{m}$. Then \eqref{eq:jointconv} holds if $n_k^{-1/2}\sum_{t=1}^{n_k}U_{i,0,k} = O_P(1)$ and $m_k^{-1/2}\sum_{t=1}^{m_k}U_{i,1,k}$ obeys a central limit theorem for $1\leq k\leq q + 1$. The same argument applies with the roles of $n_k$ and $m_k$ reversed if relatively few individuals are available pre intervention. 
 
 The calculations in the preceding two paragraphs can be generalized to include covariates and additional time periods at the expense of more involved notation and non-singularity conditions. The same types of arguments also apply if each cluster consists of one or few units over many time periods, although the conditions for time dependence are generally less involved. See \citet{dedeckeretal2007} for a comprehensive overview.  These remarks and the calculations in this example also apply to the regression model in Example \ref{ex:clusterreg}. \sqed	
\end{example}

\begin{remark}[Nonlinear models]
The methodology presented here also includes nonlinear models because the parameter $\delta$ does not need to be interpretable by itself. For example, suppose the model in Example \ref{ex:clusterreg} is the latent model in a binary choice framework with symmetric link function $F$ and $\beta_k \equiv \beta$. Then $F(\theta_0 + \delta + \beta' x) - F(\theta_0 + \beta' x)$ for some $x$ may be the treatment effect of interest but $H_0\colon \delta = 0$ still determines whether the treatment effect is zero or not. Estimates of $\theta_0$ and $\theta_1 = \theta_0 + \delta$ from these models typically do not have closed form in the presence of covariates but generally have asymptotic linear representations to which the same types of arguments as in Example \ref{ex:diffindiff} can be applied.\sqed
\end{remark}

Before concluding this section, I present a brief summary of how the rearrangement test can be implemented in practice. By Theorem \ref{t:consistency}, the following procedure provides an asymptotically $\alpha$-level test in the presence of a finite number of large clusters when only a single cluster received treatment. The test is computationally simple and does not require simulation or resampling, can be two-sided or one-sided in either direction, is able to detect all fixed alternatives, and is powerful against $1/\sqrt{n}$-local alternatives. Recall that $\varrho$ here measures how much more variable the estimate from the treated cluster $\smash{\hat{\theta}_1}$ can be relative to the second-least variable control cluster estimate $\hat{\theta}_{0,k}$. A $\varrho$ of $5$ means that the (asymptotic) variance of $\smash{\hat{\theta}_1}$ can be up to $5^2 = 25$ times larger. There is no restriction on how much \emph{less} variable $\smash{\hat{\theta}_1}$ can be than any of the other estimates and $\smash{\hat{\theta}_1}$ can be infinitely more variable than the least variable control cluster. (See also the discussion above Theorem \ref{t:size}.)

\begin{algorithm}[Rearrangement test]\label{a:test}
\begin{enumerate}
\item \label{a:test1} Choose $w$ from Table~\ref{tb:worstcasebounds} for the given number of control clusters $q$, desired significance level $\alpha$, and maximal tolerance for heterogeneity, e.g., $\varrho = 2$.
	\item Compute for each untreated cluster $k = 1,\dots, q$ an estimate $\hat{\theta}_{0,k}$ of $\theta_0$ and compute an estimate $\smash{\hat{\theta}_1}$ of $\theta_1$ from the treated cluster so that the difference $\theta_1-\theta_0$ is the treatment effect of interest. (See Examples \ref{ex:clusterreg} and \ref{ex:diffindiff} above.) Use $\hat{\theta}_n = (\hat{\theta}_{1}, \hat{\theta}_{0,1} \dots, \hat{\theta}_{0,q})$ as if it were $X$ in \eqref{eq:xdef} to compute $S = S(\hat{\theta}_n, w)$ with $w$ as in Step \eqref{a:test1}. Note that $\bar{X}_0$ is replaced here by $q^{-1}\sum_{k=1}^q \smash{\hat{\theta}_{0,k}}$.
\item Reorder the entries of $S$ from largest to smallest. Denote this by $S^\smx$ as defined above \eqref{eq:tdef}. Compute $T(S)$ and $T(S^\smx)$ as in \eqref{eq:tdef}.
\item Reject $H_0\colon \theta_1 = \theta_0$ in favor of
\begin{enumerate}
\item $H_1\colon \theta_1 > \theta_0$ if $T(S) = T(S^\smx)$.
\item $H_1\colon \theta_1 < \theta_0$ if $T(-S) = T((-S)^\smx)$.
\item $H_1\colon \theta_1 \neq \theta_0$ if either $T(S) = T(S^\smx)$ or $T(-S) = T((-S)^\smx)$ but use $\alpha/2$ in Step \eqref{a:test1}.\sqed
\end{enumerate}
\end{enumerate}	
\end{algorithm}

This test can also be used as a ``robustness check'' if inference was originally performed with a method designed for a finer level of clustering, e.g., at the county level instead of the state level. In that case Algorithm \ref{a:test} can illustrate how well the results of the original test hold up if there is dependence across counties. As I point out in Section \ref{s:normal}, one could start at $\varrho = 0$ or $\varrho = 1$ and increase $\varrho$ until the null hypothesis can no longer be rejected. This is informative because a result that holds up to a potentially $\varrho^2 = 25$ times larger variance is more credible than a result that only holds if $\varrho^2 = 1$, i.e., if $\hat{\theta}_1$ cannot be more variable than all but one $\hat{\theta}_{0,k}$. If the rearrangement test is used in difference-in-differences models in conjunction with the popular \citet{conleytaber2011} test, it is important to note that $\varrho^2 = 1$ still allows for substantial heterogeneity whereas the \citetalias{conleytaber2011} test presumes full homogeneity across clusters. 

An \texttt{R} command that implements Algorithm \ref{a:test} and the robustness check for any choice of $\varrho$ is available at \href{https://hgmn.github.io/rea}{\texttt{https://hgmn.github.io/rea}}. The next section shows how the rearrangement test performs in simulations and an application.

\section{Numerical results}\label{s:montecarlo}
This section explores the finite-sample behavior of the rearrangement test in two experiments. Example \ref{ex:ctcomp} compares the rearrangement test to the widely used \citet{conleytaber2011} test in the two-way fixed effects model with clusters. Example \ref{ex:tenncare} applies the rearrangement test as a robustness check for the results of \citet{garthwaiteetal2014}. The discussion focuses on one-sided tests to the right but the results apply more generally.

%\begin{example}[Comparison of means]
	
%\end{example}

\begin{example}[Two-way fixed effects; \citealp{conleytaber2011}]\label{ex:ctcomp}
This example uses a Monte Carlo experiment to compare rearrangement to the \citeauthor{conleytaber2011} (\citeyear{conleytaber2011}) test. The \citetalias{conleytaber2011} test is designed specifically for difference in differences and applies to models with a single treated cluster. Following \citet[sec.~V]{conleytaber2011}, the data are generated from the two-way fixed effects model 
\begin{equation}\label{eq:twfe}
Y_{t,k} = \delta I_t D_k + \eta_t + \zeta_k + U_{t, k},
\end{equation}
where $I_t$ is a post-intervention indicator, $D_k$ is a treatment indicator, and $\eta_t$ and $\zeta_k$ are time and cluster fixed effects, respectively. The error term satisfies 
\begin{equation}\label{eq:twfe_ar1} U_{t,k} = \gamma U_{t-1,k} + \sigma^{1\{k=q+1\}} V_{t, k},\end{equation} where the $V_{t, k}$ are iid copies of a standard normal variable and $k=q+1$ is the one cluster that received treatment. The model uses $\eta_t \equiv 0 \equiv \zeta_k$, ten time periods with four post-intervention periods, and, unless stated otherwise, $\gamma = .5$ and $\delta = 0$. I do not consider all of \citeauthor{conleytaber2011}'s variations of their model and, to focus on the simplest possible situation, I do not include covariates. I expand upon their analysis by investigating smaller numbers of control clusters $q$ and values of $\sigma$ other than one. In the latter situation, the \citetalias{conleytaber2011} test can be expected to fail because it relies heavily on homogeneity of all clusters in absence of an intervention. The \citetalias{conleytaber2011} test can be restored (as $q\to \infty$) if the exact form of heterogeneity is known \citep{fermanpinto2019, ferman2020} but this is not assumed here.

The \citetalias{conleytaber2011} test with one treated cluster can be computed as follows: (1)~Regress the outcome on $I_t D_k$, time and cluster fixed effects, and other covariates (if available). Denote the coefficient on $I_t D_k$ by $\hat\delta$. (2) Split the residuals by cluster and run, for each of the $q$ control clusters separately, regressions of the residuals on a constant and $I_t$. (3) Compute the $1-\alpha$ empirical quantile of the $q$ coefficients on $I_t$. Reject $H_0\colon \delta = 0$ if $\hat\delta$ is larger than that quantile. 

The rearrangement test can be computed similarly from $q+1$ separate artificial regressions of $Y_{t,k}$ on a constant and the post-intervention indicator $I_t$, 
\begin{align*}
Y_{t, k} &= \zeta + \theta_0 I_t + \mathrm{error}_{t,k}, \qquad 1\leq k\leq q,\\
Y_{t, k} &= \zeta + \theta_1 I_t + \mathrm{error}_{t,k}, \qquad k = q+1,
\end{align*}
where $\zeta$ is the intercept in each regression. The coefficients on the post-intervention indicator can be expressed as $\theta_0 = \bar{\eta}_+ - \bar{\eta}_-$ and $\theta_1 = \delta + \bar{\eta}_+ - \bar{\eta}_-$, where $\bar{\eta}_-$ and $\bar{\eta}_+$ are time averages of $\eta_t$ pre and post intervention, respectively. Because $\delta = \theta_1 - \theta_0$, I apply the rearrangement test to the least squares estimates $\hat{\theta}_{0,1}, \dots, \hat{\theta}_{0,q}$ and $\hat{\theta}_1$ of $\theta_0$ and $\theta_1$, respectively. I view \eqref{eq:twfe} as coming from individual-level data aggregated to the cluster level with a fixed number of time periods. The estimates $\hat{\theta}_1, \hat{\theta}_{0,1}, \dots, \hat{\theta}_{0,q}$ should therefore be approximately normal for the rearrangement test to apply. To test deviations from this assumption in finite samples, I also consider a situation where the innovations $V_{t,k}$ in \eqref{eq:twfe_ar1} are $\chi^2_2/2$ variables centered at zero. These innovations are asymmetric but still have unit variance.

\begin{figure}
\centering
\resizebox{\textwidth}{!}{
\input{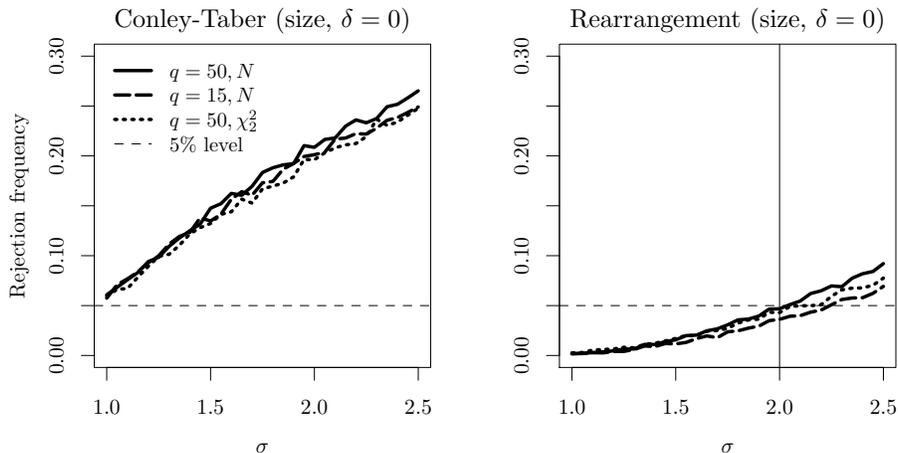}
}	
\caption{Rejection frequencies of a true null as a function of the heterogeneity $\sigma$ for the \citetalias{conleytaber2011} test (left) and the rearrangement test (right) with (i)~$q=50$ control clusters and normal errors (solid lines), (ii)~$q=15$ and normal errors (long-dashed), and (iii)~$q=50$ and chi-squared errors (dotted). The short-dashed line equals $.05$. The rearrangement test uses $\varrho = 2$ (vertical line).}\label{f:fig2}
\end{figure}

Figure~\ref{f:fig2} shows the rejection frequencies of a true null hypothesis $H_0 \colon \delta = 0$ as a function of $\sigma \in \{1, 1.05, 1.1, \dots, 2.5\} $ for the two tests at the 5\% level (short-dashed lines). The assumptions of the \citetalias{conleytaber2011} test (left) hold as $q\to \infty$ when $\sigma = 1$ but are violated at any sample size as soon as $\sigma > 1$. The rearrangement test (right) here uses $\varrho = 2$ (vertical line). The assumptions of the rearrangement test are violated as soon as $\sigma > 2$. The figure shows rejection rates in 10,000 Monte Carlo experiments for each horizontal coordinate with (i)~$q=50$ control clusters (solid lines), (ii)~$q=15$ (long-dashed), and (iii)~$q=50$ but the $V_{t,k}$ are iid copies of a $(\chi_2^2-2)/2$ variable (dotted). Both methods were faced with the same data. As can be seen, the \citetalias{conleytaber2011} test over-rejected slightly at $\sigma=1$ but quickly became unusable as $\sigma$ increased. It exceeded a 10\% rejection rate at about $\sigma = 1.25$. At $\sigma = 2.5$, the \citetalias{conleytaber2011} test falsely discovered a nonzero effect in about 25\% of all cases. In contrast, the rearrangement test was able to reject at or below the nominal level of the test as long as $\sigma \leq \varrho$. For $\sigma > \varrho$, the rearrangement test eventually started to over-reject. It performed worst at $\sigma = 2.5$, where it rejected in 6.9-9.2\% of all cases.

I also conducted a large number of additional experiments under the null. I considered (not shown) other distributions for $V_{t,k}$ and other values of the AR(1) coefficient $\gamma$, the number of time periods, the number of post-intervention periods, and the number of control clusters. However, I found that these changes had little impact on the results in the preceding paragraph. The \citetalias{conleytaber2011} test performed well when there was no heterogeneity but over-rejected wildly otherwise. More results in this direction can be found in \citet{canayetal2014}, who come to the same conclusion in their experiments. The rearrangement test continued to be highly robust to heterogeneity as long as $\varrho$ was not chosen to be much too small.

\begin{figure}
\centering
\resizebox{\textwidth}{!}{
\input{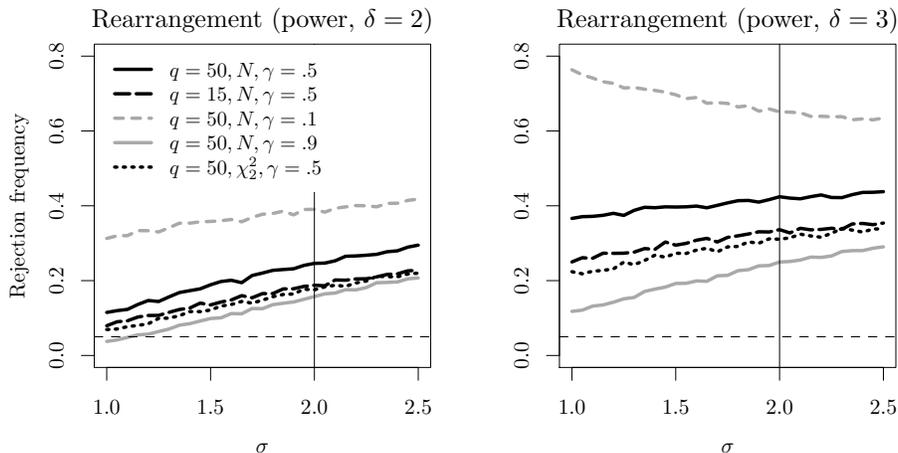}
}	
\caption{Rejection frequencies of the rearrangement test ($\varrho = 2$) under the alternative as a function of the heterogeneity $\sigma$ at $\delta = 2$ (left) and $\delta = 3$ (right) with (i) and (ii) as in Figure \ref{f:fig2}, (iii)~is (i) with weak time dependence $\gamma=.1$ (short-dashed grey), (iv)~is (i) with strong time dependence $\gamma=.9$ (solid grey) (v)~is (i) with chi-squared errors (dotted). The short-dashed line equals $.05$.}\label{f:fig3}
\end{figure}

I now turn to the performance of the rearrangement test under the alternative. The behavior of the \citetalias{conleytaber2011} test under the alternative is not discussed due to its massive size distortion. I consider the same models as before together with some variations mentioned in the preceding paragraph but use nonzero $\delta$.  Figure \ref{f:fig3} shows the results with $\delta = 2$ (left) and $\delta = 3$ (right). The base model is again model (i) with $q=50$ control clusters, standard normal $V_{t,k}$, and time dependence set to $\gamma = .5$ (solid lines). The other models deviate from (i) in the following ways: (ii)~uses $q=15$ (long-dashed), (iii)~lowers the time dependence to $\gamma = .1$ (short-dashed grey), (iv)~increases the time dependence to $\gamma = .9$ (solid grey), and (v) changes the innovations to $(\chi_2^2-2)/2$ (dotted). As can be seen, having to guard against near arbitrary heterogeneity of unknown form made it difficult to detect a relatively small treatment effect (left) when the number of control clusters was low, the distribution of the innovations was non-normal, or the treatment effect was obfuscated by strong time dependence. However, the rearrangement test reliably detected smaller treatment effects when the time dependence was relatively weak. Increasing the treatment effect (right) improved detection rates substantially and uniformly across models, with strong time dependence again being the most challenging situation. The rearrangement test now had considerable power even when only 15 control clusters were available, the innovations were asymmetric, or the time dependence was not extreme. Power was very high when there was little time dependence.  

Figures \ref{f:fig2} and \ref{f:fig3} also illustrate two noteworthy aspects of the rearrangement test: (1) The inequality the rearrangement is based on is nearly tight (as discussed below equation \eqref{eq:testfunalpha}) in the sense that it cannot be meaningfully be improved upon unless $q$ is very small. This can be seen in the right panel of Figure \ref{f:fig2}, where the rejection rate of the test was essentially at or slightly below nominal level when $\sigma = \varrho$. (2) Rejection rates under the null hypothesis increase with $\sigma$ but this does not necessarily translate into increased rejection rates under the alternative for large $\sigma$. This is seen in the right panel of Figure \ref{f:fig3}, where the power decreases with $\sigma$ in the presence of weak time dependence ($\gamma = .1$). \sqed
\end{example}

\begin{example}[Health insurance and labor supply; \citealp{garthwaiteetal2014}]\label{ex:tenncare}
In this example, I use the rearrangement test to reanalyze the results of \citet{garthwaiteetal2014}. They use a difference-in-differences design to study the effects of a large-scale disruption of public heath insurance on labor supply. Their design exploits that in 2005 approximately 170,000 adults in Tennessee (roughly 4\% of the state's non-elderly, adult population) abruptly lost access to TennCare, the state's public health insurance system. \citeauthor{garthwaiteetal2014}\ use data from the 2001-2008 March Current Population Survey to determine health insurance and work status for the years 2000-2007. The comparison groups for Tennessee are the 16 other Southern states\footnote{The Southern states are Alabama, Arkansas, Delaware, the District of Columbia, Florida, Georgia, Kentucky, Louisiana, Maryland, Mississippi, North Carolina, Oklahoma, Tennessee, Texas, Virginia, South Carolina, and West Virginia.} defined by the U.S.\ Census Bureau. 

The main treatment effect in \citet[their $\beta$ in their equation (1)]{garthwaiteetal2014}\ can be estimated as $\delta$ in \[ Y_{t, k} = \theta_0 I_t + \delta I_t D_k + \zeta_k + U_{t,k}, \] where $Y_{t,k}$ is a state-by-year mean of an outcome of interest for state $k$ in year $t$, $I_t = 1\{t\geq 2006\}$ is a post-intervention indicator, and $D_k$ equals one for an observation from Tennessee and equals zero otherwise. There are $17\times 8 = 136$ state-by-year means in total. \citeauthor{garthwaiteetal2014}\ estimate the model in the preceding display by least squares and conduct inference about $\delta$ with bootstrap standard errors that are compared to Student $t$ critical values with 16 degrees of freedom. Their preferred bootstrap first draws states with replacement and then draws individuals within those states with replacement. This type of inference accounts for autocorrelation within individuals over time but generally requires the number of clusters to be infinite for the asymptotics. This bootstrap also does not account for potential dependence within states. 

\begin{table}\caption{Effects of TennCare disenrollment in \citet[Table II.A]{garthwaiteetal2014} with their auto-correlation robust bootstrap standard errors (top) and the largest $\varrho^2$ at which a rearrangement test robust to arbitrary correlation within states and over time still detects an effect (bottom).}\label{t:ggn}
{\centering
\scalebox{.92}{
\begin{tabular}{cp{0cm}cccccc}	
\hline
& & (1) & (2) & (3) & (4) & (5) & (6) 
\\ \cline{3-8}
& & 			& 			&Employed 	 &Employed 		 &Employed 		&Employed\\
& &Has public 	& 			&working & working 	 & working 	& working\\
& &health 		& 			&$<$20 hours&$\geq$20 hours&20-35 hours	&$\geq$35 hours \\
& &insurance 	&Employed 	&per week	 &per week	 	 &per week	 	&per week \\ \cline{3-8}
$\hat{\delta}$ &&$-$0.046\phantom{$-$} &0.025 &$-$0.001\phantom{$-$} &0.026 &0.001 &0.025 \\
s.e.&&(0.010) &(0.011) &(0.004) &(0.010) &(0.007) &(0.011)\\
$p$-val.&&[0.000] &[0.019] &[0.621] &[0.011] &[0.453] &[0.020]
\\ \\
& &\multicolumn{6}{c}{Rearrangement test: largest $\varrho^2$ at which $H_0\colon \delta = 0$ is rejected}\\
 $\alpha$ & &\multicolumn{6}{c}{(``$\times$'' indicates that $H_0\colon \delta = 0$ cannot be rejected for any $\varrho \geq 0$)} \\ \cline{1-1}\cline{3-8}
$.10$ & &5.434 &1.793 &$\times$ &2.208  &$\times$   &$\times$ \\
$.05$ & &2.914 &0.972 &$\times$ &1.195   &$\times$   &$\times$ \\ \hline
\end{tabular}	
}}
\end{table}

I replicate the findings of \citet{garthwaiteetal2014} in the top panel of Table~\ref{t:ggn}. They estimate the causal effect of the TennCare disenrollment on the probability of (1) having public health insurance, (2) being employed, and (3)-(6) being employed for a certain number of hours per week. I show their bootstrap standard errors in parentheses but report one-sided $p$-values in brackets instead of their two-sided $p$-values. In (1) the alternative is a negative effect, for (2)-(6) the alternative is positive. \citeauthor{garthwaiteetal2014}\ find a highly significant 4.6 percentage point decrease for (1) and mostly significant positive effects for (2)-(6). They document an approximately 2.5 percentage point increase in employment and find the same effect if the outcome is restricted to individuals working more than 20 hours or more than 35 hours a week. All three effects are significant at the 5\% level. The inference in \citeauthor{garthwaiteetal2014}\ shows no significant effect for individuals working less than 20 hours or 20-35 hours. 

I now apply the rearrangement test as a robustness check. I view each state over time as a single cluster and run 17 separate least squares regressions of the form 
\begin{align*}
Y_{t, k} &= \theta_0 I_t + \zeta_k + U_{t,k}, \qquad 1\leq k\leq 16,\\
Y_{t, k} &= \theta_1 I_t + \zeta_k + U_{t,k}, \qquad k = 17,
\end{align*}
to obtain $\hat{\theta}_{0,k}$ ($1\leq k\leq 16$) from each of the Southern states except Tennessee and $\hat{\theta}_1$ from Tennessee ($k=17$). Note that the $\zeta_k$ are now the constant terms in each regression. To perform the robustness check, I start with $\varrho = 0$ and increase $\varrho$ by $.001$ in Algorithm~\ref{a:test} as long as the null hypothesis $H_0\colon\delta = 0$ is still rejected. The bottom panel of Table~\ref{t:ggn} shows the largest feasible value of $\varrho^2$ for outcomes (1)-(6). At the 10\% level, the result in (1) survives an up to 5.4 times larger variance in the estimate from Tennessee relative to the second-least variable control cluster estimate. The result in (2) holds if Tennessee has a 1.8 times larger variance and (4) holds even with an up to 2.2 times larger variance. At the 5\% level, these three results remain valid with smaller $\varrho^2$ but the result in (2) only survives if the estimate from Tennessee is at most slightly less variable than the second-least variable control cluster estimate. The results in (3) and (5) confirm findings in \citet{garthwaiteetal2014} in that they are not significant at any level and for any value of $\varrho$.

A noteworthy situation occurs in (6), where the rearrangement test disagrees sharply with the significant effect found by \citet{garthwaiteetal2014}. The rearrangement test finds no effect at any significance level and for any $\varrho$. In contrast, the effects in (2) and (6) are not only essentially identical but also have identical standard errors. (The $p$-values differ slightly because of rounding.) This also illustrates that the rearrangement test differs fundamentally from inference based on $t$ statistics and resampling. 

In sum, the rearrangement test robustly confirms---with one exception---the results of \citet{garthwaiteetal2014}. There is statistical evidence of increased employment concentrated among individuals working at least 20 hours per week even if one accounts for arbitrary dependence within states and over time. The results hold up to substantial heterogeneity across clusters even if the number of clusters is treated a fixed for the analysis. It is also worth noting that $\varrho$ only restricts heterogeneity in one direction. All of the results presented here are robust to arbitrary heterogeneity in any other direction and to Tennessee being infinitely more variable than the least variable control cluster. \sqed

%(Recall that $\varrho$ measures how much more variable the estimate from Tennessee can be relative to the second-least variable control cluster estimate.)

\end{example}

\section{Conclusion}\label{s:conc}
I introduce a generic method for inference about a scalar parameter in research designs with a finite number of large, heterogeneous clusters where only a single cluster received treatment. This situation is commonplace in difference-in-differences estimation but the test developed here applies more generally. I show that the test asymptotically controls size and has power in a setting where the number of observations within each cluster is large but the number of clusters is fixed. The test combines independent, approximately Gaussian parameter estimates from each cluster with a weighting scheme and a rearrangement procedure to obtain its critical values. The weights needed for most empirically relevant situations are tabulated in the paper. The critical values are computationally simple and do not require simulation or resampling. The test is highly robust to situations where some clusters are much more variable than others. Examples and an empirical application are provided. 

\appendix

\section{Proofs}
\begin{proof}[Proof of Theorem \ref{t:size}]
Choose any $\lambda\in\Lambda$ and $w\in (0,1)$. Let $S(X, w) = S = (S_1,\dots,S_{q+2})$. By continuity, we have $T(S) = T(S^\smx)$ if and only if $S_1 + S_2= S_{(q+2)} + S_{(q+1)}$ and $\sum_{k=1}^{q} S_{k+2} = \sum_{k=1}^{q} S_{(k)}$ almost surely. Conclude that \[\ev_{\lambda, 0} \varphi (X,w) =  \prob_{\lambda, 0} \Bigl(\min\{ (1+w)(X_1 - \bar{X}_0), (1-w)(X_1 - \bar{X}_0)\} > \max_k (X_{0,k}-\bar{X}_0) \Bigr).\] 

Because of the centering, we can without loss of generality assume $\mu_0 = 0$. Define $X_{1,1} = (1+w)X_1$ and $X_{1,2} = (1-w)X_1$. Use monotonicity of maximum and minimum to express the right-hand side of the preceding display as $\prob_{\lambda, 0}(\min\{ X_{1,1} - w\bar{X}_0, X_{1,2} + w\bar{X}_0 \} > X_{0,(q)}).$ Let $s^2 = \sum_{k=1}^q \sigma_k^2$ and %\[ \mu \mapsto X(\mu) = \bigl((1+w)(Y-\mu),(1-w)(Y-\mu), Z_1 - \mu,\dots, Z_q - \mu\bigr). \] 
denote by $\tilde{\varphi}(X,w)$ an infeasible version of the test function $\varphi (X,w)$ that replaces $\bar{X}_0$ by $\mu_0$.
The inequality $|1\{ a > b \} - 1\{ c > b \}|\leq 1\{ |a-b| \leq |a-c| \} $ for $a,b,c\in\mathbb{R}$ and the triangle inequality then imply that for every $t > 0$ \[ \sup_{\lambda\in\Lambda} \bigl|\ev_{\lambda, 0} \varphi (X, w)1\{ |\bar{X}_0| \leq st \} - \ev_{\lambda, 0} \tilde{\varphi} (X, w)1\{ |\bar{X}_0| \leq s t \}  \bigr| \] cannot exceed 
\begin{align*} 
\sup_{\lambda\in\Lambda} \prob_{\lambda, 0} \bigl(|X_{1,(1)} - X_{0,(q)}|\leq | \min\{ X_{1,1} - w\bar{X}_0, X_{1,2} + w\bar{X}_0\} - X_{1,(1)}|, |\bar{X}_0| \leq st\bigr).
\end{align*}
By monotonicity, this is at most $\sup_{\lambda\in\Lambda} \prob_{\lambda, 0} (|X_{1,(1)} - X_{0,(q)}| \leq w s t)$. Note that $X_{1,(1)}$ is negatively skewed and $X_{0,(q)}$ positively skewed. Because $X_{1,(1)}$ and $X_{0,(q)}$ are independent, $\prob_{\lambda, 0} (|X_{1,(1)} - X_{0,(q)}| \leq w s t)$ is largest when $X_{1,(1)}$ has the least skew. This happens at $\sigma = 0$ and implies \[ \sup_{\lambda\in\Lambda}\prob_{\lambda, 0} (|X_{1,(1)} - X_{0,(q)}|  \leq ws t) = \sup_{\lambda\in\Lambda}\prob_{\lambda, 0} (|X_{0,(q)}|  \leq w s t). \]
The probability on the right is the supremum of $\prod_{k=1}^q \Phi( w s t/\sigma_k) - \prod_{k=1}^q \Phi(-w s t/\sigma_k)$ over $\lambda\in\Lambda$.
Because $s/\sigma_k$ is decreasing in $\sigma_k$, the entire expression must be decreasing in $\sigma_k$ and the supremum in the preceding display is therefore attained at $\sigma_1 = \dots = \sigma_{q-1} = \ubar{\sigma}$ and $\sigma_q = 0$. Conclude that $\sup_{\lambda \in \Lambda} \prob_{\lambda, 0} (|X_{1,(1)} - X_{0,(q)}| \leq ws t) \leq \Phi(\sqrt{q-1} w t )^{q-1}$. Because
\[ \bigl|\ev_{\lambda, 0} \varphi (X,w)1\{ |\bar{X}_0| > st \} - \ev_{\lambda, 0} \tilde{\varphi} (X, w)1\{ |\bar{X}_0| > s t \}  \bigr| \leq \prob(|\bar{X}_0| >  s t) = 2\Phi(- qt ) \]
and because all bounds so far are valid for every $t$, it follows that \[ \sup_{\lambda\in\Lambda} \bigl|\ev_{\lambda, 0} \varphi (X,w) - \ev_{\lambda, 0} \tilde{\varphi} (X, w)  \bigr| \leq \min_{t > 0} \Bigl ( \Phi\bigl( \sqrt{q-1} w t \bigr)^{q-1} + 2\Phi(- qt ) \Bigr).\]
 
Now consider $\ev_{\lambda, 0} \tilde{\varphi} (X, w) = \prob_{\lambda, 0} (X_{1,(1)} > X_{0,(q)} )$, which can be expressed as
\begin{align*}
 \prob\bigl((1-w)X_1 > X_{0,(q)}, X_1 > 0\bigr) + \prob\bigl((1+w)X_1 > X_{0,(q)}, X_1 < 0\bigr).
\end{align*}
The second term on the right is at most $\prob(X_{0,(q)} < 0, Y < 0) = \Phi(0)^{q+1} = 2^{-q-1}$. Use independence to write the first term of the preceding display as 
\begin{align*}
\int_{0}^{\infty} \prod_{k=1}^q \Phi\biggl(\frac{(1-w)  \sigma y}{\sigma_k}\biggr) \phi(y) dy \leq \int_{0}^{\infty} \Phi\biggl(\frac{(1-w)  \bar{\sigma} y}{\ubar{\sigma}}\biggr)^{q-1} \phi(y) dy,
\end{align*}
where the inequality follows because the the integrand is increasing in $\sigma$, decreasing in $\sigma_k$, and at most one $\sigma_k$ can be arbitrarily close to zero. %At $w = 1$, the integral on the right equals $\Phi(0)^{q-1}/2 = 2^{-q}$. 
Combine the bounds on $\ev_{\lambda, 0} \tilde{\varphi} (X,w)$ and $\ev_{\lambda, 0} \varphi (X,w) - \ev_{\lambda,0} \tilde{\varphi} (X,w)$ to obtain the bound $\xi_q$.

Now consider the alternative. We still have \[\ev_{\lambda, \delta} \varphi (X, w) =  \prob_{\lambda, \delta} \Bigl(\min\{ (1+w)(X_1 - \bar{X}_0), (1-w)(X_1 - \bar{X}_0)\} > \max_k (X_{0,k}-\bar{X}_0) \Bigr).\] Because $1\{\min\{ (1+w)(X_1 - \bar{X}_0), (1-w)(X_1 - \bar{X}_0)\} > \max_k (X_{0,k}-\bar{X}_0)\}\to 1$ almost surely as $\delta \to \infty$ for $w\in (0, 1)$, dominated convergence implies $\ev_{\lambda,\delta} \varphi (X,w)\to 1$. At $w=1$,  $\min\{ 2(X_1 - \bar{X}_0), 0\} - \max_k (X_{0,k}-\bar{X}_0) \to - \max_k (X_{0,k}-\bar{X}_0)$ almost surely as $\delta \to \infty$. This limit has a continuous distribution function at $0$. At $w=1$, the Slutsky lemma implies that the preceding display converges to $\prob(0 > \max_k (X_{0,k}-\bar{X}_0)) = \prob (\bar{X}_0 > \max_k X_{0,k}) = 0$, as required.
\end{proof}

\begin{proof}[Proof of Proposition \ref{p:lowerbound}]
Let $A_t  = \bigcap_{k=1}^q \{ -t < X_{0,k} \leq t \}$ for some $t > 0$. As above, assume without loss of generality that $\mu_0 = 0$ and recall that $\ev_{\lambda, \delta} \varphi (X,w) =  \prob_{\lambda, \delta} (\min\{ X_{1,1} - w\bar{X}_0, X_{1,2} + w\bar{X}_0 \} > X_{0,(q)})$. For every fixed $t$, this is strictly larger than
\[\prob \bigl(\min\{ X_{1,1} - w\bar{X}_0, X_{1,2} + w\bar{X}_0 \} > X_{0,(q)}, A_t \bigr)\geq \prob \bigl(\min\{ X_{1,1}, X_{1,2}\} - w t > t, A_t \bigr)\] because $X_{0,(q)}\leq t$ and $|\bar{X}_0|\leq t$. By independence and because $t > 0$, the display can be expressed as \[ \prob_{\lambda, \delta} \biggl(X_1 > \frac{1+w}{1-w}t\biggr)\prob_{\lambda}(A_t) = \prob_{\lambda, \delta} \biggl(X_1 > \frac{1+w}{1-w}t\biggr) \prod_{k=1}^q\bigl(\Phi(t/\sigma_k) - \Phi(-t/\sigma_k)\bigr). \]
By symmetry, this simplifies to
\[ \Phi\Biggl(\biggl(\frac{1+w}{1-w}t - \delta\biggr)/\sigma\Biggr) 2^q\prod_{k=1}^q\bigl(\Phi(t/\sigma_k) - 0.5\bigr) \]
and, because $t$ was arbitrary, it must be true that \[ \ev_{\lambda, \delta} \varphi (X,w) \geq  2^q\sup_{t \geq 0}\Phi\Biggl(\biggl(\delta - \frac{1+w}{1-w}t\biggr)/\sigma\Biggr)\prod_{k=1}^q\bigl(\Phi(t/\sigma_k) - 0.5\bigr). \]
Replace $t$ by $t\sigma$ to obtain the bound in the proposition.

The quantity inside the supremum is continuous on $[0,\infty]$, equals zero at $t=0$ and $t=\infty$, and is strictly positive on $t\in (0,1)$. The space $[0,\infty]$ with the order topology is compact and the supremum must therefore be attained on $t\in (0,\infty)$ to not contradict the extreme value theorem. The supremum in the preceding display is therefore a maximum over $t \in (0,\infty)$ for every fixed $\delta \in [0,\infty)$ and the maximized function is a continuous function of $\delta$ on $[0,\infty]$ by the Berge maximum theorem. As $\delta \to \infty$, the supremum is attained at $t = \infty$ and the right-hand side of the display equals one.
\end{proof}

\begin{proof}[Proof of Proposition \ref{t:approx}]
Let $S(X_n, w) = S_n = (S_{1,n}, \dots, S_{q+2, n})$. We cannot have \[ \min\{S_{1,n}, S_{2,n}\} < \max\{S_{3,n},\dots, S_{q+2,n}\} \] and $T(S_n) = T(S_n^\smx)$ at the same time. Moreover, the reverse inequality implies $T(S_n) = T(S_n^\smx)$. Conclude that 
\begin{align*}
\ev \varphi (X_n, w) &=  \prob\bigl( \min\{S_{1,n}, S_{2,n}\} > \max\{S_{3,n},\dots, S_{q+2,n}\}\bigr)\\
&\qquad + \prob\bigl(T(S_n) = T(S_n^\smx), \min\{S_{1,n}, S_{2,n}\} = \max\{S_{3,n},\dots, S_{q+2,n}\}\bigr).
\end{align*}
By the assumed weak convergence and the continuous mapping theorem, we have $S(X_n, w) \leadsto S(X, w) = (S_1,\dots, S_{q+2})$. Use the continuous mapping theorem again to deduce \[\min\{S_{1,n}, S_{2,n}\} - \max\{S_{3,n},\dots, S_{q+2,n}\}\leadsto \min\{S_{1}, S_{2}\} - \max\{S_{3},\dots, S_{q+2}\}.\] The right-hand side can be expressed as \[ h_{X_{0,1},\dots, X_{0,q}}(X_1) \coloneqq \min\{ (1+w)(X_1 - \bar{X}_0), (1-w)(X_1 - \bar{X}_0)\} - \max_k (X_{0,k}-\bar{X}_0), \] where $x\mapsto h_{X_{0,1},\dots, X_{0,q}}(x)$ is strictly increasing and continuous for almost every realization of $X_{0,1},\dots, X_{0,q}$ and therefore has a strictly increasing and continuous inverse $h_{X_{0,1},\dots, X_{0,q}}^{-1}$ almost everywhere. Independence implies that the distribution function of the preceding display equals $x\mapsto \ev \Phi(h_{X_{0,1},\dots, X_{0,q}}^{-1}(x)/\sigma)$, which is continuous by dominated convergence. Conclude that $h_{X_{0,1},\dots, X_{0,q}}(X_1)$ must have a continuous distribution function at $0$ so that \[ \prob (\min\{S_{1,n}, S_{2,n}\} - \max\{S_{3,n},\dots, S_{q+2,n}\} > 0)   \to \ev \varphi (X,w)\] and $\prob(\min\{S_{1,n}, S_{2,n}\} - \max\{S_{3,n},\dots, S_{q+2,n}\} = 0) \to 0$. Combine these two results to obtain $\ev \varphi (X_n,w)\to \ev \varphi (X,w) + 0$, as desired.
\end{proof}

\begin{proof}[Proof of Theorem \ref{t:consistency}]
Let $X_{1,n} = \sqrt{n}(\hat{\theta}_{1} - \theta_1)$ and  $X_{0,k,n} = \sqrt{n}(\hat{\theta}_{0,k} - \theta_0)$ for $1\leq k\leq q$. By assumption, $X_n = (X_{1,n}, X_{0,1,n}, \dots, X_{0,q,n})\leadsto X$. Because $x\mapsto \varphi_\alpha(x)$ is invariant to multiplication of $x$ with positive constants, we have $\varphi_\alpha(\hat{\theta}_n) = \varphi_\alpha(X_n)$ if $\theta_1 = \theta_0$. By Proposition \ref{t:approx} and Theorem \ref{t:size}, this implies $\ev \varphi_\alpha(\hat{\theta}_n)\to \ev \varphi_\alpha(X) \leq \alpha$ under the null hypothesis.

Suppose $\theta_1 = \theta_0 + \delta/\sqrt{n}$. Let $x \mapsto S_\alpha(x) = S(x, w_q(\alpha,\varrho))$ and $\Delta = (\delta1\{k = 1\})_{1\leq k\leq {q+1}}$. By the assumed continuity and the Slutsky lemma, we have $X_n + \Delta \wto{\theta} X + \Delta$. Because $\sqrt{n}S_\alpha(\hat{\theta}_n) = S_\alpha(X_n+\Delta)$ and $\varphi_\alpha$ is invariant to scaling of $S$ by positive constants, it follows from Proposition \ref{t:approx} that $\ev \varphi_\alpha(\hat{\theta}_n)$ that $\ev \varphi_\alpha(\hat{\theta}_n) = \ev \varphi_\alpha(X_n + \Delta) \to \ev \varphi_\alpha(X + \Delta)$, to which the lower bound developed in Proposition \ref{p:lowerbound} can be applied.

Now suppose $\delta = \theta_1 - \theta_0 > 0$. Let $\bar{X}_{0,n} = q^{-1}\sum_{k=1}^q X_{0,k,n}$. Because $X_n/\sqrt{n}\leadsto 0$, the continuous mapping theorem implies that \[ \min\{ (1+w)(X_{1,n} + \delta -\bar{X}_{0,n}),(1-w)(X_{1,n} + \delta-\bar{X}_{0,n}) \} - \max_k (X_{0,k,n} - \bar{X}_{0,n}) \] divided by $\sqrt{n}$ converges weakly to $\min\{ (1+w)\delta,(1-w)\delta \}$. Because zero is a continuity point of the distribution of this degenerate variable unless $\delta = 0$, conclude that $\ev \varphi_\alpha(\hat{\theta}_n)\to 1$ by the same arguments as in Propsition \ref{t:approx}.
\end{proof}

\bibliographystyle{chicago}
\bibliography{qspec.bib}

\end{document}